\documentclass[10pt,journal,compsoc]{IEEEtran}
\ifCLASSOPTIONcompsoc
  \usepackage[nocompress]{cite}
\else
  \usepackage{cite}
\fi
\ifCLASSINFOpdf
\usepackage[pdftex]{graphicx}
\else
\usepackage[dvips]{graphicx}
\fi
\usepackage{amsmath}
\usepackage{url}

\usepackage{amssymb}
\usepackage{amsfonts}
\usepackage{amsthm}
\usepackage{framed}

\usepackage[ruled,vlined,linesnumbered]{algorithm2e}

\usepackage{xcolor}

\def\proof{{{\em Proof:} \/}}
\def\qed{{$\hfill \Box$\\}}

\newcommand {\ol}[1]{\overline{#1}}

\newcommand{\spell}{\mathsf{sp}}
\newcommand{\ed}{\mathsf{ed}}
\newcommand{\as}{\mathsf{as}}

\newcommand{\gap}{\text{`$-$'}\xspace}

\newtheorem{Teo}{Theorem}
\newtheorem{Lem}[Teo]{Lemma}

\newcommand{\ignore}[1]{}

\newcommand{\alex}[1]{\textcolor{blue}{#1}}

\begin{document}
%
\title{Hardness of Covering Alignment:\\ Phase Transition in Post-Sequence Genomics}

\author{Romeo~Rizzi, Massimo~Cairo, Veli~M\"akinen, Alexandru~I.~Tomescu and Daniel Valenzuela%
\IEEEcompsocitemizethanks{\IEEEcompsocthanksitem R. Rizzi is with the Department of Computer Science, University of Verona, Italy. E-mail: romeo.rizzi@univr.it\protect\\
M. Cairo is with the University of Trento, Italy. E-mail: massimo.cairo@unitn.it\protect\\
V. M\"akinen, A.I. Tomescu and D. Valenzuela are with the Helsinki Institute for Information Technology HIIT, Department of Computer Science, University of Helsinki, Finland. E-mail: \{veli.makinen,alexandru.tomescu,daniel.valenzuela\}@helsinki.fi}
\thanks{Manuscript received XXX; revised XXX.}}

%
%

\markboth{}{Rizzi \MakeLowercase{\textit{et al.}}: Hardness of Covering Alignment: Phase Transition in Post-Sequence Genomics}

\IEEEtitleabstractindextext{%
\begin{abstract}
Covering alignment problems arise from recent developments in genomics; so called pan-genome graphs are replacing reference genomes, and advances in haplotyping enable full content of diploid genomes to be used as basis of sequence analysis. In this paper, we show that the computational complexity will change for natural extensions of alignments to pan-genome representations and to diploid genomes. More broadly, our approach can also be seen as a minimal extension of sequence alignment to labelled directed acyclic graphs (labeled DAGs). Namely, we show that finding a \emph{covering alignment} of two labeled DAGs is NP-hard even on binary alphabets. A covering alignment asks for two paths $R_1$ (red) and $G_1$ (green) in DAG $D_1$ and two paths $R_2$ (red) and $G_2$ (green) in DAG $D_2$ that cover the nodes of the graphs and maximize the sum of the global alignment scores: $\as(\spell(R_1),\spell(R_2))+\as(\spell(G_1),\spell(G_2))$, where $\spell(P)$ is the concatenation of labels on the path $P$. Pair-wise alignment of haplotype sequences forming a diploid chromosome can be converted to a two-path coverable labelled DAG, and then the covering alignment models the similarity of two diploids over arbitrary recombinations.  We also give a reduction to the other direction, to show that such a recombination-oblivious diploid alignment is NP-hard on alphabets of size $3$.      
\end{abstract}

}

\maketitle

\IEEEdisplaynontitleabstractindextext

%
\IEEEpeerreviewmaketitle

\IEEEraisesectionheading{\section{Introduction}\label{sec:introduction}}

%
%
%

\IEEEPARstart{F}{or} decades, sequence alignments have played a central role in computational molecular biology and especially in computational genomics. Interestingly, being a fundamental computer science problem, there has been a constant interplay with developments in theoretical computer science forums around the problem, and the development of practical bioinformatics tools. Most notably, this connection is visible in the so-called read aligners \cite{bowtie,bwa} that use Burrows-Wheeler indexing techniques \cite{FM05,GV05}. A recent breakthrough \cite{BI15} connects the difficulty of finding sub-quadratic time solution to pair-wise alignment to a  complexity theory question. There are still open questions around the basic sequence alignment setting (e.g. indexed approximate pattern matching), but at the same time the computational genomics community is moving towards abstractions beyond sequences, where even the most fundamental questions are open. One of the latest trends is to replace a reference genome with a \emph{pan-genome variant graph} \cite{Maretal16}, with a backbone consisting of a reference sequence and alternative paths encoding common variants observed in a population. 
A basic theoretical question and decisive technological issue is how the role of the sequence alignment toolbox and conceptual framework can scale up in elaborating this more structured data, a world intrinsically populated by labelled directed graphs, that in many cases we can assume to be acyclic at least to a large extent (labelled DAGs) (see e.g.~\cite{Maretal16}). 
One possible formulation is to ask for the minimum number of edits to convert one graph to another; this is MAX SNP-hard even when the input DAGs are unordered trees \cite{AroraLMSS98}. 

In this paper, we refine the tractability border of alignments by focusing on graphs that are as close to sequences as possible. Namely, we focus on labelled DAGs that are coverable by two paths. A covering alignment asks for two paths $R_1$ (red) and $G_1$ (green) in DAG $D_1$ and two paths $R_2$ (red) and $G_2$ (green) in DAG $D_2$ that cover the nodes of the graphs and maximize the sum of the global alignment scores: \[\as(\spell(R_1),\spell(R_2))+\as(\spell(G_1),\spell(G_2)),\] where $\spell(P)$ is the concatenation of labels on the path $P$. We show that this problem is NP-hard even on binary alphabets. A more principled way to derive this formulation comes from modeling \emph{diploid} genomes \cite{MV14}, where the labelled DAG is a grid graph denoting a pair-wise alignment of haplotypes. We defer the detailed derivation and applications of this natural similarity measure in the context of diploid alignments to Section~\ref{sec:diploidalign}. We show that this restricted variant of covering alignment, called \emph{recombination-oblivious diploid alignment problem}, is NP-hard on alphabets of size $3$. This problem becomes polynomial time solvable once one of the input alignments needs not be covered by the optimal solution, or when the problem is otherwise similarly relaxed \cite{KG98,MV14,MV15} (see Section~\ref{sec:diploidalign}).

We hope these results are starting points for a more systematic study of sequence analysis in the era of post-sequence genomics: Our reduction from multiple alignment to covering alignment of two labelled DAGs would seem to indicate that problems that are NP-hard on many sequences become NP-hard already on two inputs with higher level abstractions. As the reduction works on a binary alphabet, on DAGs minimally harder than sequences, and natural relaxations of the problem become solvable in polynomial time, we have thus identified a phase transition between polynomially-solvable and NP-hard alignment problems. 

Our reductions follow a general approach introduced in~\cite{Rizzi2013} to show the NP-completeness of the problem of deciding whether a string is a square. 

\section{Preliminaries}

\subsection{Problem definition}

Let $\Sigma$ be a finite alphabet. We use $\Sigma^*$ to denote the set of all strings over $\Sigma$ and $\Sigma^+$ to denote the set of all not-empty strings over $\Sigma$. In this paper we will also use the term \emph{sequence} to denote a string. The empty string is denoted by $\varepsilon$, and $\Sigma_\varepsilon$ denotes $\Sigma \cup \{\varepsilon\}$. For a string $A=a_1 a_2 \cdots a_\ell$, $|A|$ denotes its length, namely $\ell$, and $A[j]$ denotes its $j$th symbol, namely $a_j$. We will use the indexed product notation $\prod$ to denote repeated concatenation of strings. The \emph{edit distance} of strings $A$ and $B$, denoted $\ed(A,B)$, is defined as the minimum number of deletions, insertions and substitution of symbols to convert $A$ into $B$. 

For $\Sigma' \in \{\Sigma, \Sigma_\varepsilon, \Sigma^*, \Sigma^+\}$,
a \emph{$\Sigma'$-DAG} is a tuple $\mathcal{D} = (D,\ell)$, where $D=(V,E)$ is a DAG with $|V|$ nodes and $|E|$ edges, and $\ell:V\mapsto \Sigma'$ is a total function on $V$. For a path $P=v_1, \ldots, v_t$ in $D$, the \emph{spelling} of $P$, denoted $\spell(P)$, equals $\ell(v_1) \cdots \ell(v_t)$. We say that a set $\mathcal{P}$ of paths in $D$ is a \emph{path cover} if every node in $V$ appears in some $P \in \mathcal{P}$. The \emph{width} of $D$ equals the minimum cardinality of a path cover of $D$.

For $\Sigma' \in \{\Sigma, \Sigma_\varepsilon, \Sigma^*, \Sigma^+\}$ we consider the following problem:

\begin{framed}
\noindent \textbf{Path covers of minimum edit distance in two $\Sigma'$-DAGs ({\sc PC-Min-ED-$\Sigma'$})}\\
\noindent INPUT: Two $\Sigma'$-DAGs $\mathcal{D}_1 = (D_1, \ell_1)$ and $\mathcal{D}_2 = (D_2, \ell_2)$ of width 2.\\
\noindent OUTPUT: A path cover $\{R_1,G_1\}$ of $D_1$ and 
        a path cover $\{R_2,G_2\}$ of $D_2$
        minimizing \[\ed(\spell(R_1),\spell(R_2)) + \ed(\spell(G_1),\spell(G_2)).\]
\end{framed}

Here $R$ works as an analogy to \emph{red} path and $G$ works as an analogy to \emph{green} path.

\ignore{
\begin{sloppypar}
\alex{[Probably remove for this submission]} As we discuss here below, the {\sc PC-Min-ED-$\Sigma'$} problem has many natural variants, and our NP-hardness proof pertains to all of these. A more general objective function could be $\alpha_R \ed(\spell(R_1),\spell(R_2)) + \alpha_G \ed(\spell(G_1),\spell(G_2))$ for some positive real constants $\alpha_R$ and $\alpha_G$ with $\alpha_R+\alpha_G=1$. Another objective function could be $\max \{\ed(\spell(R_1),\spell(R_2)), \ed(\spell(G_1),\spell(G_2))\}$, or the problem could ask to lexicographically minimize the vector $[\ed(\spell(R_1),\spell(R_2)), \ed(\spell(G_1),\spell(G_2))]$.
\end{sloppypar}}



\subsection{Edit distance and optimal alignments}

The notion of edit distance is tightly connected with that of a pair-wise alignment (see e.g. \cite{MBCT15} for an introduction to these notions). A \emph{pair-wise alignment} of two sequences $A,B \in \Sigma^*$ is a pair of strings $(A',B')$ where:
\begin{itemize}
\item $A',B' \in (\Sigma \cup \{\gap\})^*$, where \gap is a special \emph{gap symbol};
\item $A',B'$ have the same length $L$;
\item each $A'$ and $B'$ contains exactly $L-|A|$ and $L-|B|$ gap symbols, respectively.
\end{itemize}
Thus, $A$ and $B$ are \emph{subsequences} of $A'$ and $B'$, respectively, in the sense that they can be obtained from them by deleting zero or more symbols. A pair $(A'[i],B'[i])$ is called 
\begin{itemize}
\item \emph{identity}, if $A'[i],B'[i] \in \Sigma$ and $A'[i] = B'[i]$,
\item \emph{substitution}, if $A'[i],B'[i] \in \Sigma$ and $A'[i] \neq B'[i]$,
\item \emph{deletion}, if $A'[i] \in \Sigma$ and $B'[i] = \gap$,
\item \emph{insertion}, if $B'[i] \in \Sigma$ and $A'[i] = \gap$.
\end{itemize}
An insertion or deletion is also called a \emph{gap}. The set of all pair-wise alignments of $A$ and $B$ is denoted by $\mathsf{A}(A,B)$. The edit distance of $A$ and $B$ can also be expressed in terms of alignments, as
\[\ed(A,B) = \min_{(A',B') \in \mathsf{A}(A,B)}  |\{i \in \{1,\dots,|A'|\}\;:\; A'[i]\neq B'[i]\}|\]

Given a \emph{scoring function} $s : \Sigma \cup \{\gap\} \mapsto \mathbb{R}$, the \emph{global alignment score} of a pairwise alignment $(A',B')$ is \[\as(A',B') = \sum_{i \in \{1,\dots,|A'|\}} s(A'[i],B'[i]).\] An \emph{optimal alignment} of $A$ and $B$ is an alignment of maximum global alignment score. With the scoring scheme $s(\gap,c)=s(c,\gap)=-1$, $s(a,b)=-1$ for all $a\neq b \in \Sigma$, and $s(a,a)=0$, for all $a \in \Sigma$, finding the optimal global alignment score of $A$ and $B$ is equivalent to computing their edit distance. Unless otherwise stated, in the rest of this paper we assume that an ``optimal alignment'' refers to this scoring scheme for edit distance.


In Section~\ref{sect:starhardness}, we prove that the
 {\sc Min-ED-2PC-$\Sigma_\varepsilon$} problem
 (and hence the {\sc Min-ED-2PC-$\Sigma^*$} problem) is NP-hard in all of the above variants. Remarkably, these negative results hold also in the case of a binary alphabet $\Sigma := \{0,1\}$. The instances resulting from the reduction can also be cast as inputs to the Recombination-Oblivious Diploid Alignment Problem (see Section~\ref{sec:diploidalign}); the two problems are polynomially equivalent on these instances and this proves that Recombination-Oblivious Diploid Alignment Problem is also NP-hard. 

\subsection{Further notations for strings and graphs}

A string $S$ over $\Sigma$ of length $n$ can be expressed as a $\Sigma$-DAG $\overline{S}$ with $n$ nodes and of width~$1$ consisting of a single path $P = v_1, v_2, \dots, v_n$ with $\ell(v_i) = S[i]$ (equivalently, $\spell(P) = S$). 

\ignore{Moreover, let $\overrightarrow{S}$ denote the transitive closure of $\overline{S}$, that is, the DAG $\overrightarrow{S}$ having  
$V(\overrightarrow{S}) = V(\overline{S})$, the very same labeling $\ell : V(\overrightarrow{S}) = V(\overline{S}) \mapsto \Sigma$, but $A(\overrightarrow{S}) = \{(v_i,v_j) \; : \; i < j \}$.
Note that both $\overline{S}$ and $\overrightarrow{S}$
are $\Sigma$-DAGs with a single source and a single sink.
\alex{Is the transitive closure needed?}}

Let $\mathcal{D}_1 = (D_1, \ell_1)$ be a $\Sigma'$-DAG with a single sink $t_1$
and $\mathcal{D}_2 = (D_2, \ell_2)$ be a $\Sigma'$-DAG with a single source $s_2$.
The $\Sigma'$-DAG obtained by adding the arc $(t_1,s_2)$ to the disjoint union of $\mathcal{D}_1$ and $\mathcal{D}_2$ is denoted by $\mathcal{D}_1\mathcal{D}_2$, juxtaposing the aliases, just as with strings, to suggest the concatenation in series of the actual objects.


Given a $\Sigma^*$-DAG $\mathcal{D}$, we denote by $\mathcal{D}_\varepsilon$ the $\Sigma_\varepsilon$-DAG obtained from $\mathcal{D}$ by expanding nodes labeled by strings of length more than $1$ into paths. Namely, each node $v$ labelled by $S$, $|S| \geq 2$ is replaced by the path $\ol{S}$; the arcs incident to $v$ get updated as follows: the arcs entering (exiting, resp.) $v$ now enter (exit, resp.) the first (the last, resp.) node of $\ol{S}$.

Analogously, given a $\Sigma^+$-DAG $\mathcal{D}$, we denote by $\mathcal{D}_\Sigma$ the $\Sigma$-DAG obtained from $\mathcal{D}$ in the same manner as above.

\ignore{ 
In Section~\ref{sect:genhardness}, we refine the construction
given in Section~\ref{sect:starhardness}
to obtain the stronger result that the {\sc PC-Min-ED-$\Sigma$} problem is also NP-hard in all of these variants.

In an extended version of this paper, we will refine the construction
given in Section~\ref{sect:starhardness}
to obtain the stronger result that the {\sc PC-Min-ED-$\Sigma$} problem is also NP-hard in all of these variants. 
}

\section{NP-hardness of {\sc PC-Min-ED-$\Sigma_\varepsilon$}}
\label{sect:starhardness}

In this section, the NP-hardness of {\sc PC-Min-ED-$\Sigma'$}
is shown for the case in which the empty string can occur as a label for some of the nodes, i.e., the labeling function is not total on $V$. 
\ignore{
The NP-hardness proof for the restricted problem where the $\varepsilon$ label is banned will be obtained in a later section,
by refining the construction first given in this one.
This subsequent restructuring will add one further layer on the top of the basic reduction. The direct integration of this layer, while possible,
would severely obscure the main principles behind the general working of the reduction.\\
}

Let $\mathbf{N}_n := \{0,1,\ldots, n-1\}$. For brevity, we also denote the binary $\bmod$ operation by $\%$. The reduction, which we will describe in Section~\ref{subsec:reduction}, is from the following problem:\\

\begin{framed}
\noindent \textbf{Longest Common Subsequence} ({\sc LCS})\\
\noindent INPUT: $n$ strings $S_0, \ldots, S_{n-1}$.\\
\noindent OUTPUT: A longest possible string $S$
              that is a subsequence of every $S_i$, $i\in \mathbf{N}_n$.
\end{framed}

{\sc LCS} is known to be NP-hard even when the input strings are all binary and of the same length $\ell$~\cite{Mai78}. Moreover, we can assume that each $S_i$ contains both a $0$ and a $1$. Given $n$ input strings $S_0, \dots, S_{n-1}$ of the same length $\ell$ to the LCS problem, we show how to construct
two $\Sigma^*$-DAGs $\mathcal{A}$ and $\mathcal{B}$ of width $2$ such that the following two lemmas hold.

\begin{Lem} \label{lem:NPCstar_easy}
   Let $S'$ be a common subsequence for $S_0, \ldots , S_{n-1}$,
   and let $\delta = \ell - |S'|$.
   Then there exists a path cover $\{A_r,A_g\}$ of $\mathcal{A}_\varepsilon$, and a path cover $\{B_r,B_g\}$ of $\mathcal{B}_\varepsilon$, such that $\ed(\spell(A_r),\spell(B_r)) = 0$ and $\ed(\spell(A_g),\spell(B_g)) = 2\,\delta$. Hence, $\ed(\spell(A_r),\spell(B_r))+\ed(\spell(A_g),\spell(B_g)) = 2\,\delta$.
\end{Lem}

\begin{Lem} \label{lem:NPCstar_hard}
   Let $\{A_r,A_g\}$ be a path cover for $\mathcal{A}_\varepsilon$, and let $\{B_r,B_g\}$ be a path cover for $\mathcal{B}_\varepsilon$. Let $d:= \ed(\spell(A_r),\spell(B_r)) + \ed(\spell(A_g),\spell(B_g))$.  
   Then there exists a common subsequence $S'$ for $S_0, \ldots , S_{n-1}$
   with $d \geq 2(\ell - |S'|)$.
\end{Lem}

As the reader will check, the construction can be easily performed in polynomial time. As a consequence, the above two lemmas (whose formal proofs will be given later, after describing the construction) will prove the NP-hardness
of {\sc PC-Min-ED-$\Sigma_\varepsilon$}.\ignore{ in essentially all of the variants introduced.}
\ignore{(Only minor modifications will also settle the variants requiring to minimize the functional $\max \{d(\spell(R_1),\spell(R_2)), d(\spell(G_1),\spell(G_2))\}$.) }
\ignore{
The extension of these negative results to $\Sigma$-DAGs
will be discussed in Subsection~\ref{subsec:not-empty} 
}

\subsection{The reduction, and the general idea behind it}
\label{subsec:reduction}

Let $S_0,\dots,S_{n-1}$ be $n$ binary strings of the same length $\ell$, each having both a $0$ and a $1$.

In the reduction, we will use an integer $M$ that will play the role of a sufficiently big
constant. A string $T$ whose length depends on $M$ will play the role of a firm \emph{tab} gadget,
capable of forcing an optimal alignment to align the $i$-th occurrence of $T$ in one string to the $i$-th occurrence of $T$ in the other string. We now explain how to choose $T$.

A \emph{linear de Bruijn sequence of order $k$} over a binary alphabet is a string in which every binary string of length $k$ appears as substring exactly once~\cite{FSM94,deB46}. Let $DB(k)$ denote one such string. The string $DB(k)$ has length $2^k + k - 1$ and can be constructed in linear time by taking the spelling of an Eulerian cycle in a de Bruijn graph of order $k-1$ \cite{Goo46,BP07}.

\ignore{\begin{Lem}
\label{lem:random-no-repeats}
Let $S$ be a random $\{0,1\}$-string
of fixed length $|S|$ and let $\ell=O(\log |S|)$. Then, with high probability, $S$
has no repeated substring of length $\ell$, i.e., for any $1\leq i,j\leq |S|-\ell$,
we have $S[i..i+\ell-1]=S[j..j+\ell-1]$ iff $i=j$.
\end{Lem}
\begin{proof}
Take $i\neq j$. We have $\Pr(S[i..i+\ell-1]=S[j..j+\ell-1])=2^{-\ell}$,
since the events $S[i+\delta]=S[j+\delta]$ for $0\leq\delta\leq\ell-1$
are independent of probability $1/2$. Applying the union bound we
get 
\[
\Pr(S[i..i+\ell-1]=S[j..j+\ell-1]\text{ for some }i\neq j)\leq n^{2}\cdot2^{-\ell}\leq n^{2}\cdot2^{-\alpha\log n}=n^{2-\alpha}.
\]
\qed
\end{proof}
}

\begin{figure}
\centering
\includegraphics[width=4cm]{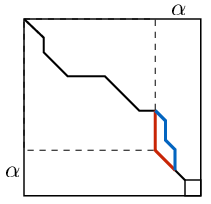}
\caption{Under the unit cost edit distance, if the compared strings have a common suffix of length $|\alpha|$, the end of any optimal alignment (marked with blue in the $\alpha$ region of identities) can be canonicalized so that it consists first of a sequence of insertions/deletions and then a sequence of identities (marked with red). Symmetric canonicalization can be done for a common prefix.\label{fig:whollyaligned}}
\end{figure}

\begin{Lem}
\begin{sloppypar}
Let $\alpha_{1},\dots,\alpha_{q}$ and $\beta_{1},\dots,\beta_{q}$ be strings of length at most $M$. Let $k$ be such that $|DB(k)|=\Theta(qM\log qM+qM^{2})$ and let $T = DB(k)$. Then the two strings 
\end{sloppypar}
\[A:=\alpha_{1}T\alpha_{2}T\dots\alpha_{q-1}T\alpha_{q}\]
\[B:=\beta_{1}T\beta_{2}T\dots\beta_{q-1}T\beta_{q}\]
have an optimal alignment
that aligns perfectly the $q-1$ occurrences of $T$ in each string.
\label{lem:separator}
\end{Lem}
\begin{proof}
Take an optimal alignment and suppose that the $k$-th character of
the $i$-th occurrence of $T$ in $A$ is aligned with the same $k$-th
character of the $j$-th occurrence of $T$ in $B$. Then, it can be
assumed that these occurrences of $T$ are wholly aligned, without losing
optimality (see Figure~\ref{fig:whollyaligned}). Hence, it is sufficient to rule out any optimal alignment
where some occurrence of $T$ in $A$ has \emph{no} character aligned
with any other occurrence of $T$ in $B$. We show that such an
alignment has cost $\omega(qM)$, so it is worse than aligning \emph{only}
the $q-1$ occurrences of $T$, thus it is not optimal.

Suppose by contradiction that the $i$-th occurrence of $T$ in $A$
(denoted with $T_{i}$) is such that:
\begin{itemize}
\item for no $1\leq k\leq|T|$ and $1\leq j\leq q$, the $k$-th character
of $T_{i}$ is aligned with the $k$-th character the $j$-th occurrence
of $T$ in $B$,
\item the cost of aligning $T_{i}$ with the smallest substring of $B$
containing the aligned characters (denoted with $B'$) is $o(qM)$.
\end{itemize}
Observe that $T_{i}$ is aligned by identities with at least one substring
$B''$ of $B'$ of size 
\begin{align*}
|T|/o(qM) & =\omega(|T|/qM)=\omega(qM\log qM/qM+qM^{2}/qM)\\
		  & =\omega(\log qM+M)=\omega(M+\log|T|).
\end{align*}
(See Figure~\ref{fig:pigeonhole} for the reasoning.)

\begin{figure}
\centering
\includegraphics[width=6cm]{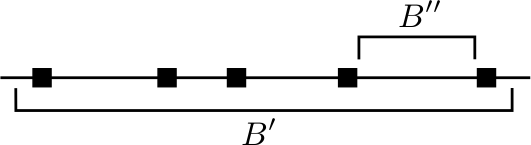}
\caption{Black boxes indicate substitutions or gaps; even if they are evenly distributed, there is a long region $B''$ of identities.\label{fig:pigeonhole}}
\end{figure}

This substring may include up to $M$ characters from
some $\beta_{h}$, but then it includes at least $\omega(\log|T|)=\omega(k)$
consecutive characters from an occurrence of $T$ in $B$, contradicting
$T$ being a de Bruijn sequence (after fixing suitable constants in the asymptotic notation). \qed
\end{proof}
  
The high-level structure of the two $\Sigma^*$-DAGs $\mathcal{A}$ and $\mathcal{B}$ is depicted in Figure~\ref{fig:structAB}. The value $N$, which we choose to be $2n\ell$, plays again the role of a sufficiently big number. The strings $T_1, T_2, \ldots, T_{N+1}$ are just identical copies of the tab gadget $T$, their subscripts are there only to indicate their position in $\Sigma^*$-DAG; we will refer to this subscript as \emph{depth}. Figure~\ref{fig:gadgetD} defines the content of the $D(i)$ gadget, for $i\in \mathbf{N}_n$. 

In our reduction, $M \leq \ell$ will suffice, and we will always apply Lemma~\ref{lem:separator} for $q = N$ or $q = N+1$ strings. Thus the tab gadget $T$ will be of length \[\Theta(N\ell\log N\ell + N\ell^2) = \Theta(n\ell^2\log n\ell^2 + n\ell^3).\]

These two $\Sigma^*$-DAGs, and their expanded versions $\mathcal{A}_\varepsilon$ and $\mathcal{B}_\varepsilon$, can be clearly constructed in polynomial time.



\begin{figure*}[h]
\centering
  \includegraphics[width=13cm]{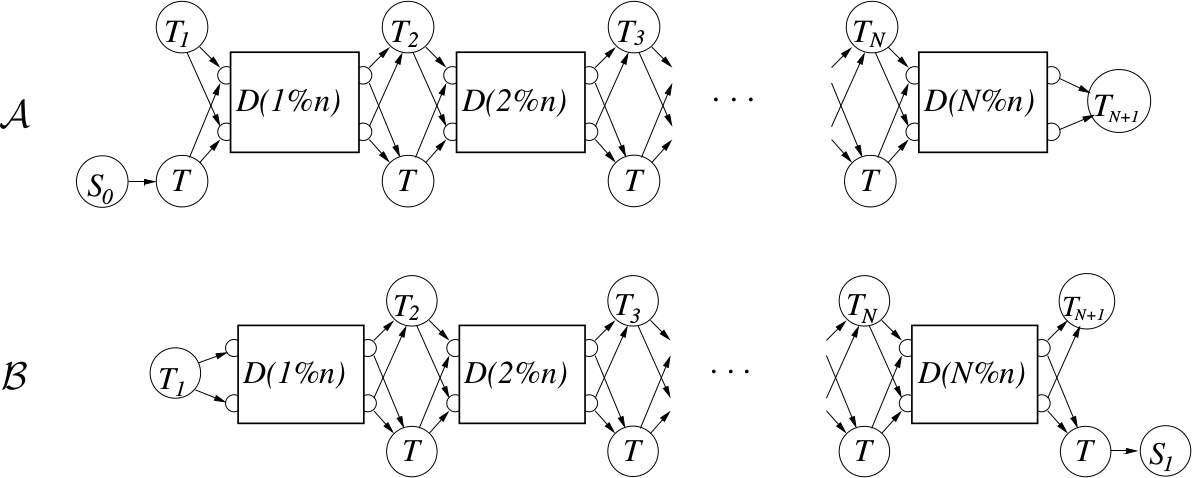}
  \caption{The high-level structure of $\mathcal{A}$ and $\mathcal{B}$.}
  \label{fig:structAB}
\end{figure*}

%
%


\begin{figure}[htb]   
\begin{center}
  \includegraphics[width=\columnwidth]{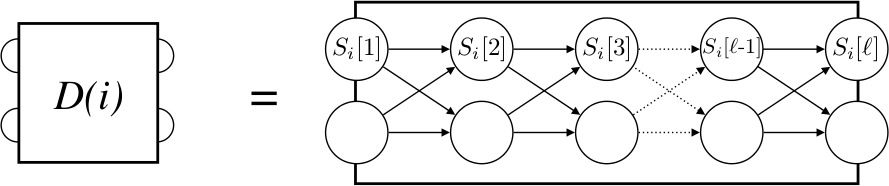}
  \caption{The $D(i)$ gadget. The empty nodes are labelled with the empty string.}
  \label{fig:gadgetD}
\end{center}
\end{figure}

\ignore{
\alex{[Probably remove]A first lower-bound on $M$, namely $M > 2\ell$,
comes most natural after considering the statements of Lemmas~\ref{lem:NPCstar_easy} and~\ref{lem:NPCstar_hard}.
A second and last lower-bound on $M$, namely $M \geq 4\ell^2$,
comes after considering that any path
entirely contained within a $D(i)$ gadget has length less than $4\ell^2$. Thus we set $M := \max\{2\ell, 4\ell^2\} = 4\ell^2$.}}

We next give the proofs of the lemmas and then extend the reduction to diploid alignments.

\subsection{Proofs of the lemmas}

The proofs depict green and red paths through the gadgets, as well as their alignments. Figures~\ref{fig:high-level-rg}~and~\ref{fig:gadgetD-rg} help to follow the constructions.

\begin{figure*}[h]
\centering
  \includegraphics[width=13cm]{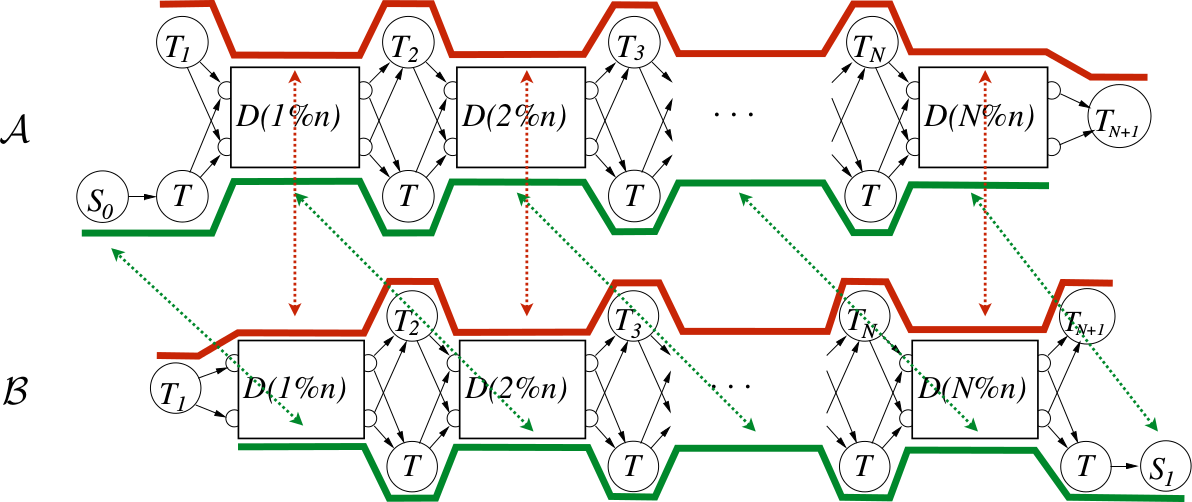}
  \caption{Red and green paths through $\mathcal{A}$ and $\mathcal{B}$ as depicted in the proofs.}
  \label{fig:high-level-rg}
\end{figure*}

\begin{figure}[htb]   
\begin{center}
  \includegraphics[width=\columnwidth]{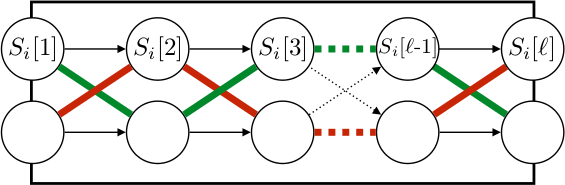}
  \caption{Red and green paths through the $D(i)$ gadget as depicted in the proofs.}
  \label{fig:gadgetD-rg}
\end{center}
\end{figure}

\noindent\emph{Proof of Lemma~\ref{lem:NPCstar_easy}:\/}
For $i\in \mathbf{N}_n$, let $S'_i$ denote the subsequence of $S_i$ obtained by deleting the symbols selected by its subsequence $S'$ (i.e., $S'_i$ is the complement of $S'$ in $S_i$). Assume that we underline in green the $|S'|$ symbols in $S_i$ which originate from $S'$ and cross out in red the $|S'_i|$ symbols in $S_i$ which originate from $S'_i$.

   Also, 
   if the $j$-th symbol of $S_i$ is underlined in green,
   then let $\psi_i[j] := \varepsilon$,
   otherwise, if the $j$-th symbol of $S_i$ is crossed out in red,
   then $\psi_i[j] := S_i[j]$.    
   Notice that there exist two (disjoint) paths $R_i$
   and $G_i$ forming a path cover of the $\Sigma^*$-DAG $D(i)$
   and such that
\[
   \spell(R_i) = \psi_i[1] \psi_i[2] \cdots \psi_i[\ell] \text{~~~and~~~} \spell(G_i) = S'.
\]

   The reader should now check that $\mathcal{A}_\varepsilon$ is jointly covered by two disjoint paths
   $A_r$ and $A_g$ such that:
\[
\spell(A_r) = 
       \left(
          \prod_{i=1}^N (T ~ \spell(R_{i \% n}))
       \right) T, \\
\]

\begin{align*}
\spell(A_g) &=
       S_0 
          \prod_{i=1}^N (T~ \spell(G_{i \% n}))
       =  
       S_0 
          \prod_{i=1}^N (T ~ S')
        \\
       & = S_0 
       \left(
          \prod_{i=1}^{N-1}(T ~ S')
       \right)  ~T ~ S'.
\end{align*}

The reader is also invited to check that $\mathcal{B}_\varepsilon$
is jointly covered by two disjoint paths
$B_r$ and $B_g$ such that:

\begin{align*}
\spell(B_r) & =
       \left(
          \prod_{i=1}^{N} (T ~ \spell(R_{i \% n}))
       \right) T 
       =
       T
       \left(
          \prod_{i=1}^{N} (\spell(R_{i \% n})~T)
       \right) \\
       & = \spell(A_r),
\end{align*}
\begin{align*}
\spell(B_g) &=
       \left(
          \prod_{i=1}^{N} (\spell(G_{i \% n}) T)
       \right)
              S_1 
       =
       \left(
          \prod_{i=1}^{N} (S'~T)
       \right)
              S_1 \\
       &=
       S'
       \left(
          T~\prod_{i=1}^{N-1} (S'~T)
       \right)
              S_1 
       =
       S'
       \left(
          \prod_{i=1}^{N-1} (T~S')
       \right) ~T
               S_1. \\
\end{align*}

Clearly, $\ed(\spell(A_r),\spell(B_r)) = 0$. By applying Lemma~\ref{lem:separator} to the strings 
\begin{align*}
& \alpha_1 = S_0, \alpha_2 = \cdots = \alpha_{N+1} = S' \\
& \beta_1 = \cdots = \beta_{N}=S', \beta_{N+1} = S_1
\end{align*}
we have that any optimal alignment of $\spell(A_g)$ and $\spell(B_g)$ perfectly aligns the $N$ occurrences of $T$. Thus: \label{alphbabeta}
\[
   \ed(\spell(A_g),\spell(B_g)) = \ed(S_0,S') + \ed(S',S_1) = \delta+\delta = 2\, \delta.
\]
\qed

To prove Lemma~\ref{lem:NPCstar_hard}, we introduce one more notation related to subgraphs. 
   For a string $S$, an $S$\emph{-subpath} of a $\Sigma'$-DAG $\mathcal{D}$
   is a $\Sigma'$-DAG $P$ such that $P$ is a subgraph of $\mathcal{D}$, it is a path and $\spell(P) = S$.  

\medskip

\noindent\emph{Proof of Lemma~\ref{lem:NPCstar_hard}:\/}
   We assume $d < 2\ell$ since otherwise the thesis holds vacuously.

Since $\mathcal{A}$ has two sources, namely $T_1$ and $S_0$, we have that each of $A_r$ and $A_g$ starts in precisely one of them. To simplify notation in what follows, let now $A_r$ denote that path starting in $T_1$ (and thus let $A_g$ be that path starting in $S_0$).
   
   Notice that $\mathcal{A}_\varepsilon$ ($\mathcal{B}_\varepsilon$) contains precisely
   $2N+1$ $T$-subpaths (which we also call \emph{tab subpaths}),
   and these are displaced as follows. 
   For $i=1,\ldots,N$, $\mathcal{A}_\varepsilon$ ($\mathcal{B}_\varepsilon$, resp.)
contains two parallel tab subpaths at depth~$i$ (at depth~$i+1$, resp.) and precisely one tab subpath at depth~$N+1$ (at depth~$1$, resp.). The idea here is that within $\mathcal{A}_\varepsilon$ (or $\mathcal{B}_\varepsilon$) we can reach the nodes in a tab subpath at depth~$i$ from the nodes in a tab subpath at depth~$i-1$. Clearly, once a solution path of $\mathcal{A}_\varepsilon$ (or $\mathcal{B}_\varepsilon$) passes through the first and the last node of a tab subpath, it traverses it entirely, holding it as a subpath of itself.   

Notice that each one of the paths $A_g$ and $A_r$ ($B_g$ and $B_r$, resp.)
must necessarily traverse precisely one tab subpath from any pair of parallel tab subpaths, i.e., precisely one tab subpath of depth~$i$, for $i=1,2,\ldots, N$ (for $i=2,3,\ldots, N+1$, resp.).

Also, 
at least one among $A_g$ and $A_r$ ($B_g$ and $B_r$, resp.) 
also traverses the single  tab subpath of depth~$N+1$ (of depth~$1$, resp.).
We claim that in fact, precisely one among $A_g$ and $A_r$ ($B_g$ and $B_r$, resp.) 
also traverses the single  tab subpath of depth~$N+1$ (of depth~$1$, resp.).

Indeed, assume for a contradiction that $A_r$ ends immediately before the tab subpath at depth~$N+1$. This implies that $A_g$ ends with $T_{N+1}$. We claim that in this case we have $d \geq 2\ell$, contradicting the assumption made at the beginning of this proof. First, note that $A_g$ has $N + 1$ of tab subpaths. If $B_g$ had a different number of tab subpaths (i.e., $N$), then $\ed(A_g,B_g) \geq |T|$. From the choice of $T$, we have that $T \geq 2\ell$ and thus $d \geq \ed(A_g,B_g) \geq 2\ell$, which is the desired contradiction. 

We now have that $\spell(A_g)$ has $S_0$ as prefix and contains $N+1$ tab subpaths, and $B_g$ has $S_1$ as suffix and contains $N+1$ tab subpaths. By Lemma~\ref{lem:separator} we have that the $N+1$ occurrences of the tab subpath are perfectly aligned, and thus $d\geq \ed(A_g,B_g) \geq |S_0| + |S_1| = 2\ell$, again the desired contradiction.

By a symmetric argument we obtain that also precisely one among $B_g$ and $B_r$ starts with the tab subpath of depth~$1$.

\ignore{
   In the case of the lexicographic metric, where we assume $d(\spell(A_r),\spell(B_r)) = 0$, it can be easily enforced that $\spell(A_r)$ has the tab string $T$ as a prefix. In the more difficult case where $\alpha_R= \alpha_G =1$,
   we can ensure this by possibly swapping $A_r$ and $A_g$
   (also swapping $B_r$ and $B_g$ at the same time).
   After this double swapping, it can be easily argued that also $\spell(B_r)$ has the tab string $T$ as a prefix.
   It also follows that $\spell(B_g)$ has a string $S'_0 T$ as a prefix,
   where $S'_0$ is a subsequence of $S_0$. This implies what anticipated above:
   $B_g$ does not traverse the $T$ subpath at depth~$1$ in $\mathcal{B}_\varepsilon$. And all the above arguments are perfectly symmetric.\\
   }
At this point, we summarize the situation as follows:
\begin{enumerate}
\item the tab subpaths of $A_g$ are precisely $N$, namely those at depth $1, 2, \ldots, N$;

\item the tab subpaths of $A_r$ are precisely $N+1$, namely those at depth $1, 2, \ldots, N, N+1$;

\item the tab subpaths of $B_r$ are precisely $N+1$, namely those at depth $1, 2, \ldots, N, N+1$.
   These are perfectly aligned with the $N+1$ tab subpaths of $A_r$.
   This means that, for every $i=1,\ldots, N$,
   the red subsequence of $D(i\%n)$ within $A_r$
   is aligned against the $D(i\%n)$ within $B_r$;

\item  the tab subpaths of $B_g$ are precisely $N$, namely those at depth $2, \ldots, N, N+1$.
   Notice that the $N$ tab subpaths of $B_g$ are out of phase
   with the $N$ tab subpaths of $A_g$.
   Namely, the first tab subpath of $B_g$ is at depth~$2$ and perfectly aligns with the first tab subpath of $A_g$ at depth~$1$.
   Therefore, the spelling of the green path through $D(1\%n)$ from $B_g$ gets aligned against the green path through $S_0$ from $A_g$.
   More generally, the spelling of the green path through $D((i+1)\%n)$ from $B_g$ gets aligned against the green path through $D(i\%n)$ from $A_g$.
\end{enumerate}

This misalignment of the two green strands,
while the two red strands perfectly are aligned,
is the key engine behind our reduction. We can now proceed with defining the common subsequence $S'$. 

We say that the $(d_1,d_2)$\emph{-interval} of $\mathcal{A}_\varepsilon$ ($\mathcal{B}_\varepsilon$) is the subgraph of $\mathcal{A}_\varepsilon$ ($\mathcal{B}_\varepsilon$)
induced by those nodes which can be reached by some node in a tab subpath
of depth~$d_1$ and which can reach some node in a tab subpath
of depth~$d_2$.
Since $d <2\,\ell$
, then there should exist some~$t = 1,\ldots, N$
such that,
the restriction of the paths $A_g$ and $A_r$ within the $(t,t+n)$-interval of $\mathcal{A}_\varepsilon$
are perfectly aligned (that is, perfectly identical)
to the the restriction of the path $B_g$ and to that of the path $B_r$ within the $(t,t+n)$-interval, respectively.
To see this, notice that the two alignments cover $N=2n\ell$ subgraphs and thus there must be a region of $2n\ell/d\geq 2n\ell/(2\ell)=n$ subgraphs inducing no alignment error. Call this region the \emph{identity zone}.

The existence of this identity zone allows us to define a common subsequence $S'$ to $S_0,\ldots, S_{n-1}$. Namely, the identities restricted to the content of $S_0$ and $S_1$ picked by an optimal covering alignment of $D(0)$ and $D(1)$ inside $A_g$ and $B_g$ fixes a common subsequence $S' = S_1[i_1]S_1[i_2] \cdots S_1[i_p]$ of $S_0$ and
$S_1$, and we need to show that this subsequence is common to all $S_0,\ldots, S_{n-1}$. 

Since $A_r$ must be picking in
$D(0)$ a complementary subsequence $S_0' = S_0[j_1]S_0[j_2] \cdots S_0[j_{\ell-p}]$ of $S_0$, where $i_{k_1} \neq j_{k_2}$ for
all $k_1$ and $k_2$, to guarantee $S_0$ is covered by $S' \cup S_0'$, then for an identity alignment, $B_r$ must
be picking in $D(0)$ a subsequence $S''$ of $S_0$ matching perfectly with $S_0'$. If one removes two
identical subsequences $S'$ and $S_0'$ from the same string ($S_0$), the resulting string is the same.
Hence one can modify $B_r$ to pick $S_0'$ instead of $S'$ without changing the alignment score.
The analogous modification of $B_g$ inside $D(0)$ to pick $S'$ within $D(0)$ also does not change the
score. One can continue propagating these modifications to the left and, analogously, to the
right until one has proven $S'$ to be a subsequence of all $S_0,\ldots, S_{n-1}$.

Since the identity region contains all different types of subDAG pairs, one can obtain an alignment with cost $d=2(\ell-|S'|)$ as follows. Copy the zero cost identity alignments to all places; with the same propagation argument as above, one observes that $S_0$ is aligned against $S'$ (being the prefix of $B_g$ before the first tab) and $S_1$ is aligned against $S'$ (being the suffix of $A_g$ after the last tab); all other parts of the alignments have cost zero. Since $S'$ is a subsequence of $S_0$ and of $S_1$, the optimal edits to make them match cost exactly $2(\ell-|S'|)$. On the other hand, there cannot be any better alignments: Each edit located between the identity zone and before $S_0$ is propagated as an extra symbol or missing symbol from $S'$ to the prefix of $B_g$ matched against $S_0$. In the former case, the extra symbols may improve the alignment of $S_0$ to the prefix of $B_g$, but these improvements cancel out with the cost of introducing these edits in the first place. In the latter case, the missing symbols just increase the cost. The case of edits between identity zone and $S_1$ is analogous. 
\ignore{[Assume now there are $2^D$s in the gadgets in place of $0^D$.] Since the identity region contains all different types of subDAG pairs, one can obtain an alignment with cost $d=2(\ell-|S'|)$ as follows. Copy the zero cost identity alignments to all places, and with the same propagation argument as above, one observes that $S_0$ is aligned against a sequence $S^p$ containing $S'$ as subsequence with other content consisting of symbols $2$ (being the prefix of $\mathcal{B}_g$ before the first tab) and $S_1$ is aligned againts against a sequence $S^s$ containing $S'$ as subsequence with other content consisting of symbols $2$ (being the suffix of $\mathcal{A}_g$ after the last tab); all other parts of the alignments have cost zero. Since $S'$ is a subsequence of $S_0$ and of $S_1$, and $S_0$ and $S_1$ contain no symbols $2$, the optimal edits to make $S^p$ and $S^s$ match them $S_0$ and $S_1$, respectively, cost exactly $2(\ell-|S'|)$. On the other hand, there cannot be any better alignments: Each edit located between the identity zone and before $S_0$ is propagated as an extra symbol or missing symbol from $S'$ to $S^p$. In the former case, the extra symbols may improve the alignment of $S_0$ to $S^p$, \emph{but these improvements cancel out with the cost of introducing these edits in the first place}. In the latter case, the missing symbols just increase the cost. The case of edits between identity zone and $S_1$ is analogous.} 
\qed

As a consequence of the above two lemmas, we obtain the claimed result.

\begin{Teo}
Problem {\sc PC-Min-ED-$\Sigma_\varepsilon$} is NP-hard on a binary alphabet.
\end{Teo}

\noindent \proof
Let $S_0,\dots,S_{n-1}$ be $n$ strings of length $\ell$ for the {\sc LCS} problem. We need to decide whether there is a common subsequence $S'$ of $S_0,\dots,S_{n-1}$ such that $\ell - |S'| = \delta$, for a given $\delta$. From this input, we construct the two DAGs $\mathcal{A}$ and $\mathcal{B}$ for problem {\sc PC-Min-ED-$\Sigma_\varepsilon$}. We claim that $S_0,\dots,S_{n-1}$ and $\delta$ is a yes input for {\sc LCS} if and only if the cost of an optimal solution for problem {\sc PC-Min-ED-$\Sigma_\varepsilon$} on $\mathcal{A}$ and $\mathcal{B}$ is at most $2\delta$. The forward and reverse implications follow from Lemma~\ref{lem:NPCstar_easy} and Lemma~\ref{lem:NPCstar_hard}, respectively.
\qed

\section{Recombination-Oblivious Diploid Alignment\label{sec:diploidalign}}

Pair-wise sequence alignments have been extended to capture many biological sequence features, such as mutation biases, repeats (DNA), splicing (RNA), and alternative codons (proteins) \cite{DEKM98,Gus97}, but extensions to diploid organisms have been less common \cite{KG98,MV14,MV15}. The motivation to model diploid alignment comes from the recent developments in sequencing and in haplotyping algorithms; it can be foreseen that one day we will have reasonably accurate haplotype sequences of each of the homologous sequences forming a chromosome pair. Such a \emph{diploid chromosome} can itself be expressed as a pair-wise alignment that stores the \emph{synchronization} of their haploid sequences, that is, telling in which positions a recombination is possible.  A \emph{recombination} of a pair-wise alignment $(A'[1..L],B'[1..L])$ of strings $A$ and $B$ is \[(A'[1..i]B'[i+1..L],B'[1..i]A'[i+1..L]),\] for some $i$. We also overload the notation $\spell(\cdot)$, and denote by $\spell(A')$ the initial string $A$, that is, $\spell(A')$ is an operation removing the gap symbols \gap from $A'$. We obtain the following problem.

\begin{framed}
\noindent \textbf{Recombination-Oblivious Diploid Alignment Problem}\\
\noindent INPUT: Alignments $(A',B')$ and $(C',D')$ of strings $A$ and $B$, and $C$ and $D$, respectively.\\
\noindent OUTPUT: Alignments $(A'',B'')$ and $(C'',D'')$ resulting from a series of recombinations to $(A',B')$ and $(C',D')$, respectively, maximizing
        \[\as(\spell(A''),\spell(C'')) + \as(\spell(B''),\spell(D'')).\] 
\end{framed}

Notice that even if $(A',B')$ and $(C',D')$ represent diploid chromosomes of two siblings, their recombination patterns are independent, so the formulation gives a way to measure just the sequence similarity not penalizing on the natural recombination phenomenon. Other applications are in comparing haplotyping results between two tools even on the same data; haplotyping algorithms usually return blocks of correctly phased variants, but not on full chromosomes.  

The Recombination-Oblivious Diploid Alignment Problem was defined in \cite{MV14}, but its complexity was left open. Related notions on comparing two sequences to a third with edit distance and crossover were studied in \cite{KG98}; polynomial dynamic programming algorithms were derived, and extensions to multiple sequences were shown NP-hard. These notions and dynamic programming algorithms were further generalized in \cite{MV14,MV15}. The case where the third sequence is an alignment, and one needs find a recombination of it to minimize the sum of edit distances of the resulting haplotypes to the two other input sequences, is still polynomial time solvable \cite{MV14,MV15}. Moreover, these algorithms  extend for the case where all three inputs are alignments, but only one of them needs to be covered by the paths through the two other alignments \cite{MV14,MV15}. Complexity changes when one needs to cover more than one alignment: We have already seen the analogous result on labelled DAGs, but in the following we state this same result in the refined model of diploid alignments, which is sligthly more specific.

\ignore{Before proceeding to the NP-hardness proof, let us consider the easier direction to reduce diploid alignment to covering alignment of labelled DAGs.  Every Recombination-Oblivious Diploid Alignment Problem instance can be encoded as two $\Sigma_\varepsilon$-DAGs: For an alignment $(A',B')$, create nodes $v^A_i$ and $v^B_i$, for $1\leq i \leq |A'|$, with $\ell(v^A_i)=\varepsilon$ if $A'[i]=\gap$ otherwise $\ell(v^A_i)=A'[i]$, and with $\ell(v^B_i)=\varepsilon$ if $B'[i]=\gap$ otherwise $\ell(v^B_i)=B'[i]$. Then create arcs $(v^A_i,v^A_{i+1}), (v^A_i,v^B_{i+1}), (v^B_i,v^B_{i+1}), (v^B_i,v^A_{i+1})$ for $1\leq i < |A'|$. Finally, add source $s$ with label $\ell(s)=\varepsilon$ connecting it to 
nodes $v^A_1$ and $v^B_1$ with arcs $(s,v^A_1)$ and $(s,v^B_1)$, and add target $t$ with label $\ell(t)=\varepsilon$ connecting it from 
nodes $v^A_{|A'|}$ and $v^B_{|A'|}$ with arcs $(v^A_{|A'|},t)$ and $(v^B_1,t)$. After encoding both inputs of the Recombination-Oblivious Diploid Alignment Problem this way, as separate $\Sigma_\varepsilon$-DAGs ,the outputs of {\sc PC-Min-ED-$\Sigma_\varepsilon$} can be cast as recombinations of the pair-wise alignments in an obvious way. To the other direction the connection is more elaborate and will be detailed next.   }

\begin{Teo}
\label{cor:dalignhard}
The Recombination-Oblivious Diploid Alignment Problem is NP-hard when alphabet size is $3$ or larger.
\end{Teo}

\noindent \emph{Proof.} We use alphabet $\Sigma=\{0,1,\mathtt{t}\}$ and fix the scoring scheme $s(r,c)$ as follows:\\
\begin{center}
\begin{tabular}{l|c|c|c|c}
     $s$        &    $0$        &    $1$         &  \texttt{t}  & \gap\\
\hline
$0$            &    $0$        &   $-1$          &   $-\infty$   & $-1$\\
$1$            &   $-1$        &   $0$           &   $-\infty$   & $-1$\\
\texttt{t}   & $-\infty$   & $-\infty$    &  $0$           & $-\infty$\\        
\gap          &   $-1$        &   $-1$         & $-\infty$   & $0$ \\
\end{tabular}\\
\end{center}
Here $s(r,c)$ is given by the value at row $r$ and column $c$.

The DAGs $\mathcal{A}$ and $\mathcal{B}$ from Section~\ref{sect:starhardness} can be cast as pair-wise alignments by taking each \emph{column} of the gadgets (as in the visualization) and considering the following cases: 
\begin{enumerate}
\item[(i)] if a column contains two nodes $v$ and $w$ with the same label $T= \ell(v)=\ell(w)$, construct a block $(\mathtt{t},\mathtt{t})$ in the alignment; 
\item[(ii)] if a column contains two nodes $v$ and $w$ with one of them, say $w$, with label $\ell(w)=\varepsilon$ construct a block $(\ell(v),\gap)$ in the alignment; 
\ignore{\item[(iii)] if a column contains only one node $v$ labeled $\ell(v)=0^D$, construct a block $(\mathtt{d},\gap)$ in the alignment; }
\item[(iii)] if a column contains only one node $v$ labeled $\ell(v)= S_0$ or $\ell(v)= S_1$, construct a block $(\ell(v),\ell(v))$ in the alignment; 
\item[(iv)] if a column contains only one node $v$ labeled $\ell(v)=T$, construct a block $(\mathtt{t},\gap)$ in the alignment. 
\end{enumerate} 
Concatenating these blocks from left to right creates pair-wise alignments $(A',B')$ and $(C',D')$ corresponding to DAGs $\mathcal{A}$ and $\mathcal{B}$, respectively. The resulting pair-wise alignment $(A',B')$ is shown in Figure~\ref{fig:alignmentA_1}

\begin{figure}[htbp]
  \begin{center}
  \includegraphics[width=1\columnwidth]{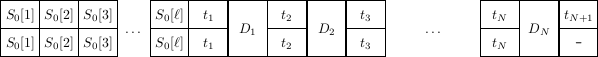}
  \end{center}
  \caption{
    High-level structure of pair-wise alignment $(A',B')$. The contents of blocks $D_i$ are shown in Figure~\ref{fig:alignmentA_2}.
    All the $t_i$ corresponds to the symbol $t$; the subscripts are to shown the relationship with the graph $\mathcal{A}$.
  }
    \label{fig:alignmentA_1}
\end{figure}

\begin{figure}[htbp]
\centering
  \includegraphics[width=0.8\columnwidth]{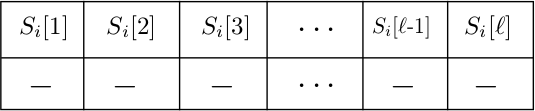}
  \caption{Pair-wise alignment version of gadget $D_i$.\label{fig:alignmentA_2}}
    
\end{figure}

Consider a series of recombinations of $(A',B')$ into $(A'',B'')$ and a series of recombinations of $(C',D')$ into $(C'',D'')$, that maximize \[\as(\spell(A''),\spell(C''))+\as(\spell(B''),\spell(D'')),\] under the scoring function define above. We claim that \[-(\as(\spell(A''),\spell(C''))+\as(\spell(B''),\spell(D'')))-2\ell\] equals the optimal solution of covering alignment of DAGs $\mathcal{A}$ and $\mathcal{B}$ with the unit cost edit distance. For the reverse implication, one can map the alignments of red and green paths in the proof of Lemma~\ref{lem:NPCstar_easy} to form alignments of $(\spell(A''),\spell(C''))$ and $(\spell(B''),\spell(D''))$, where $S_0$ and $S_1$ are deleted from the head and tail, respectively, of the alignment corresponding to red paths. Alignment corresponding to that of green paths is identical, with respect to the mapping of nodes to symbols derived above. The claimed equality then follows considering the definition of the scores. For the forward implication, since all tab symbols $\mathtt{t}$ need to align in their occurrence order as in the proof of Lemma~\ref{lem:NPCstar_hard}, and since recombinations inside the head $(S_0,S_0)$ and tail $(S_1,S_1)$ of $(A',B')$ and $(C',D')$, respectively, are non-effective, an optimal series of recombinations is in one-to-one correspondence with the covering red and green paths as in the reverse implication.      

Hence, solving Recombination-Oblivious Diploid Alignment Problem on these instances solves the {\sc PC-Min-ED-$\Sigma'$} on $\Sigma_\varepsilon$-DAGs and due to Lemmas~\ref{lem:NPCstar_easy}~and~\ref{lem:NPCstar_hard} would solve the LCS problem. 
\qed

\section{NP-hardness without empty labels\label{sect:genhardness}}

Recall that problem {\sc PC-Min-ED-$\Sigma$} differs from {\sc PC-Min-ED-$\Sigma_\varepsilon$} in that each node of the graph needs to have a non-empty label. Our plan is to modify as little as possible
the construction offered in Section~\ref{sect:starhardness} and for this purpose we consider \emph{indel} edit distance, rather than unit cost edit distance: In this scoring scheme, substitutions have cost $\infty$, indentities have cost zero, and insertions and deletions have cost $1$. We also increase the alphabet from binary to size $4$ by replacing all empty labels in $\mathcal{A}$ with a new symbol $\mathtt{a}$, and all empty labels in $\mathcal{B}$ with a new symbol $\mathtt{b}$. Obviously, any pair of covering alignments needs to have gap symbols aligned with each $\mathtt{a}$ and each $\mathtt{b}$. This cost is invariant and does not alter the relative order of alignments when sorted by their overall cost. One can thus modify systematically proofs of Lemmas~\ref{lem:NPCstar_easy}~and~\ref{lem:NPCstar_hard} taking this invariant into account to conclude that {\sc PC-Min-ED-$\Sigma$} is NP-hard with this scoring scheme and alphabet size $4$.  An analogous modification to the proof of Theorem~\ref{cor:dalignhard} gives that Recombination-Oblivious Diploid Alignment Problem is NP-hard when its input contains no gap symbols and the alphabet size is $5$. The Appendix demonstrates a subtle phase transition for this latter problem variant, as a slight relaxation of it is in P. For the interested reader, the last paragraph of the Appendix revisits the reduction to show that the derived partially covering relaxation indeed achieves better edit distance for the reduction instances than the NP-hard solution; such pair of alignments has quite a specific structure and gives also more insight to the reduction itself.  

\section{Discussion\label{sec:discussion}}

It is evident that the reductions given here generalize to scoring functions beyond those considered here. We leave such development for future work. Notice that similar fine-grained complexity analysis has been conducted for the LCS problem \cite{bonizzoni}. 

The reduction technique developed here is likely to find other applications in the area of computational pan-genomics \cite{Maretal16}. 
A direct consequence is that comparing two pan-genome representations is NP-hard, if accepting the notion of covering alignment developed here as the basis. Namely, the general optimization problem is to find minimun number $k$ of paths to cover each of the two input DAGs, and then among such covers one that maximizes the sum of $k$ global alignment scores among the $k!$ pairings. Since case $k=2$ is NP-hard, case $k=1$ is the classical quadratic time solvable sequence alignment problem, and our reduction works on binary alphabet, we have identified a phase transition for this notion of similarity (see also the Appendix for an even closer phase transition).  As the labeled DAG representation loses the connectivity information on variations in the pan-genomic setting, one could resort back to a multiple alignment of haplotypes, and adjust the notion of recombinations to allow only limited number of those. This notion allows parameterized complexity analysis. Indeed, let us consider the Recombination-Oblivious Diploid Alignment Problem from this angle. Given a limit $r$ for the number of recombinations in one alignment, a naive algorithm is to consider all $\binom{n}{r}\leq n^r$ recombinations on both input alignments and then compute the global alignment of the resulting haplotype pairs. This results into an $O(n^{2r+2})$ time algorithm. One can speed this up to $O(n^{r+3})$ by considering all recombinations only in one input alignment and then resorting to the algorithm in \cite{MV14}. We believe there is room for further work around the parameterized tractability border of this problem. For the general covering alignment problem on DAGs a plausible direction is to look for approximation algorithms or approximation hardness.

\section*{Acknowledgements} 

We thank the anonymous reviewers for good suggestions to improve the presentation.
This work was supported in part by the Academy of Finland (grants 309048, 274977).

\bibliographystyle{abbrv}
\bibliography{paper}

\begin{IEEEbiography}[]{Romeo Rizzi} received his Ph.D.~from the Department of Mathematics of Padova University, Italy, in 1997.
He held researcher positions at centers like CWI (Amsterdam, Netherlands), BRICS (Aarhus, Denmark) and IRST (Trento, Italy), University of Trento and University of Udine, Italy.
Since 2011, he has been an associate professor at the University of Verona, Italy. He has a background in Operations Research and his
main interests are in Combinatorial Optimization and Algorithms.
He is an Area Editor of 4OR. He published a hundred research papers in a broad range of scientific journals in the areas of Discrete Mathematics, Combinatorics, and Algorithms. He also authored several papers in conference proceedings, and invited chapters. Since 2004,
he has intensively acted as a trainer of the Italian team for the iOi.
\end{IEEEbiography}

\begin{IEEEbiography}[]{Massimo Cairo} is currently a PhD student in Mathematics jointly at the University of Verona and the University of Trento, Italy, working with professor Romeo Rizzi. His main research focus is in Theoretical Computer Science and Algorithms, and he has refereed publications in the fields of Computational Graph Theory, Computational Biology and Automated Temporal Planning.

\end{IEEEbiography}

\begin{IEEEbiography}[]{Veli
M\"akinen} finished his PhD studies in Computer Science in 2003 at the University of Helsinki, Finland. He worked as a Postdoctoral Researcher (2004-2005) at Bielefeld University, Germany, and then back in Helsinki as Postdoctoral Research Fellow (2005-2007) and Academy Research Fellow (2007-2010). In 2010, he was appointed as a Professor in Computer Science at the University of Helsinki. Veli M\"akinen now heads the Genome-scale algorithmics research group. His research interests are in compressed text indexing and in algorithmic bioinformatics. He has some 100 publications on these topics, including a co-authored text book.
\end{IEEEbiography}

\begin{IEEEbiography}[]{Alexandru I. Tomescu}
obtained his PhD in computer science from the University of Udine, Italy, in 2012. After spending six months at the Technical University Berlin, Germany, he joined the Genome-scale algorithmics group at the University of Helsinki, Finland, where he currently holds an Academy of Finland Postdoctoral Researcher Fellowship.
\end{IEEEbiography}

\begin{IEEEbiography}[]{Daniel Valenzuela}
obtained his MSc in Computer Science from University of Chile in 2013 and his PhD in Computer Science from University of Helsinki in 2017.
Currently he is a post-doctoral researcher at University of Helsinki, as a member of the Genome-scale algorithmics research group.
His research interests include string algorithms, compressed data structures and their applications in bioinformatics.
\end{IEEEbiography}

\vfill

\end{document}